\documentclass[10pt, a4paper, reqno]{amsart}

\usepackage{mathrsfs}
\usepackage{microtype}
\usepackage{braket}
\usepackage{tikz}
\usetikzlibrary{calc}
\usepackage[compat=newest]{yquant}
\usepackage{hyperref}
\usepackage{enumitem}

\DeclareMathOperator{\diag}{diag}
\DeclareMathOperator{\poly}{poly}

\newtheorem{theorem}{Theorem}
\newtheorem{corollary}[theorem]{Corollary}
\newtheorem{proposition}[theorem]{Proposition}
\newtheorem{lemma}[theorem]{Lemma}
\theoremstyle{definition}
\newtheorem{definition}[theorem]{Definition}
\newtheorem{remark}[theorem]{Remark}

\newcommand{\R}{\mathbf{R}}
\newcommand{\PH}{\mathsf{PH}}
\newcommand{\PP}{\mathsf{PP}}
\newcommand{\PostBQP}{\mathsf{PostBQP}}
\newcommand{\PostBPP}{\mathsf{PostBPP}}
\newcommand{\PostQAOA}{\mathsf{PostQAOA}}
\newcommand{\PostIQP}{\mathsf{PostIQP}}

\makeatletter
\let\orig@@setauthors\@setauthors
\def\@setauthors{%
  \orig@@setauthors
  {\centering\footnotesize
   \textit{Center for Quantum Computing Science, University of Latvia}\\
   \texttt{\{ralfs.abolins, andris.ambainis\}@lu.lv}\par}%
}
\makeatother

\title[A sharp interaction-degree threshold for simulating QAOA]{A sharp interaction-degree threshold \\for simulating QAOA}
\author{Ralfs Āboliņš}
\author{Andris Ambainis}

\begin{document}

\begin{abstract}
We identify a sharp interaction-degree threshold for the classical simulation of
QAOA with $2$-local cost functions.  At degree~$3$, classical sampling from
depth-$1$ QAOA, even within multiplicative error $2^{n^{s}}$ for any fixed $s <
1$, would collapse the polynomial hierarchy to its third level.  At degree~$2$,
exact classical sampling from depth-$p$ QAOA on $n$ qubits runs in polynomial
time whenever $p = O(\log n)$.  The hard degree-$3$ instances have trivially
optimizable cost functions, so sampling hardness does not by itself imply a
quantum optimization advantage.
\end{abstract}

\maketitle

\section{Introduction}\label{sec:intro}

The Quantum Approximate Optimization Algorithm, QAOA~\cite{farhi2014quantum}, is
one of the main near-term quantum algorithms.  Farhi and
Harrow~\cite{farhi2016quantum} showed that, already at depth~$1$, classical
sampling from the output distribution of QAOA with sufficiently small
multiplicative error would collapse the polynomial hierarchy.  Their result is
often cited as evidence that even very shallow QAOA can exhibit behavior
inaccessible to classical sampling.

We determine how structurally simple a QAOA instance can be while its output
distribution remains hard to sample classically.  We focus on $2$-local cost
functions $C(x_1,\ldots,x_n) = \sum_i C_i(x_1,\ldots,x_n)$, where each term
couples at most two variables.  Such functions encode standard optimization
problems such as MaxCut.  The pairwise couplings in the cost function form an
interaction graph, and a natural measure of structural complexity is the maximum
degree of this graph, which we call the interaction degree.  The construction of
Farhi and Harrow produces instances in which the interaction degree grows with
the input size, leaving open whether any fixed degree suffices for hardness.

We show that degree~$3$ already suffices.  Classical sampling from depth-$1$
QAOA of interaction degree at most~$3$ within multiplicative error $2^{n^s}$ for
any fixed $0 \leq s < 1$ would collapse the polynomial hierarchy to its third
level.  This goes well beyond the constant tolerance of earlier arguments.
Moreover, the hard instances can be chosen so that the cost function is trivial
to optimize classically, so sampling hardness alone does not yield a quantum
optimization advantage.

We then show that at interaction degree~$2$, every marginal output probability
of a depth-$p$ QAOA circuit can be computed classically in time polynomial
in~$n$ and exponential only in~$p$.  In particular, exact sampling runs in
polynomial time whenever $p = O(\log n)$.  The simulation uses the Markov--Shi
algorithm for quantum circuits of small cut width~\cite{markov2008simulating}.

A comparable dichotomy holds classically for constraint satisfaction.  SAT
restricted to instances in which every variable occurs in at most~$3$ clauses
remains $\mathsf{NP}$-complete, but becomes polynomial-time solvable when every
variable occurs in at most~$2$ clauses~\cite{tovey1984simplified}.  The
transition from 2-SAT to 3-SAT~\cite{Karp1972, GareyJohnson1979} is the most
familiar instance of the same pattern.  In the classical case the parameter
separates $\mathsf{P}$ from $\mathsf{NP}$-completeness, while for QAOA it
separates classical simulability from sampling hardness.

All simulation bounds count arithmetic operations on matrix entries.  In every
circuit of Section~\ref{sec:hardness}, and in every instance of
Section~\ref{sec:easy} with integer cost values and angles that are integer
multiples of $\pi/4$, the entries lie in $\mathbf{Q}(e^{i\pi/4})$, so every
intermediate value of the simulation stays in this field and carries rational
coefficients of bit length at most linear in the circuit size.  The same
asymptotic bounds therefore hold in the bit model, and exact sampling runs in
expected polynomial time by comparing random bits against the computable binary
expansions of the conditional probabilities, while for arbitrary real parameters
the bounds of Section~\ref{sec:easy} are understood in the real-arithmetic
model.

\section{Preliminaries}\label{sec:prelim}

QAOA takes as input a cost function $C(x_1,\ldots,x_n)$ on $n$ binary variables
together with a decomposition $C = \sum_i C_i$ into terms each depending on at
most two variables.  As a diagonal operator on $n$ qubits, $C$ acts by $C\ket{z}
= C(z)\ket{z}$.  The algorithm constructs the depth-$p$ state
\[
\ket{\pmb{\gamma}, \pmb{\beta}}
= e^{-i\beta_p B}\, e^{-i\gamma_p C}\, \cdots\,
  e^{-i\beta_1 B}\, e^{-i\gamma_1 C}\,
  \ket{+}^{\otimes n},
\]
where $B = \sum_{j=1}^n X_j$ is the mixing operator and $\pmb{\gamma},
\pmb{\beta} \in \R^p$ are variational parameters.  The state is measured in the
computational basis.  Since each $C_i$ acts on at most two qubits,
$e^{-i\gamma_k C}$ factors into commuting one- and two-qubit diagonal gates, and
$e^{-i\beta_k B}$ into single-qubit rotations.  Each $2$-local term coupling
variables~$j$ and~$k$ produces a two-qubit diagonal gate on those qubits, and
the set of all such couplings forms a graph.

\begin{definition}\label{def:igraph}
Given a cost function expressed as a sum of $2$-local terms $C = \sum_i C_i$,
the \emph{interaction graph} of this decomposition is the simple graph on vertex
set~$[n]$ with an edge $\{j,k\}$ whenever some term $C_i$ depends nontrivially
on both variables~$j$ and~$k$, meaning that flipping either bit can change the
value of $C_i$.  The \emph{interaction degree} is the maximum vertex degree of
this graph.  
We write
\emph{degree} for interaction degree when the meaning is clear.
\end{definition}


Postselection is the ability to condition on an event in a computation that may
occur with exponentially small probability and is therefore not physically
realizable in polynomial time.  In quantum computation, the conditioning event
is a measurement outcome on a designated register, and the resulting connection
to counting complexity is the identity $\PostBQP = \PP$.  The quantum class
$\PostBQP$ conditions on a designated register returning all zeros in a quantum
circuit, while its classical counterpart $\PostBPP$ conditions on a designated
bit of a probabilistic algorithm.  Both definitions require the conditioning
event to have positive probability so that the conditional distribution is
well-defined.  This probability may be exponentially small, so implementing
post-selection by sampling until the conditioning event occurs would take
exponential time.

\begin{definition}\label{def:postbqp}
A language $L$ is in $\PostBQP$ if there exists a polynomial-time uniform family
of quantum circuits, each acting on $\ket{w}\ket{0^{m(|w|)}}$ and equipped with
a designated output qubit and a post-selection register, such that for every
input~$w$:
\begin{enumerate}[label=(\roman*)]
\item the probability of the post-selection register returning $\ket{0\cdots 0}$
      is nonzero;
\item if $w \in L$, the output qubit is~$1$ with probability at least $2/3$,
      conditioned on post-selection;
\item if $w \notin L$, the output qubit is~$0$ with probability at least $2/3$,
      conditioned on post-selection.
\end{enumerate}
\end{definition}

The resulting output distribution of a postselected computation is the original output
distribution conditioned on the postselection event.  Its probabilities are
therefore ratios of event probabilities, so the error notion used in the
sampling argument must preserve such ratios.

\begin{definition}\label{def:multerr}
A classical algorithm samples with \emph{multiplicative error} $c \geq 1$ from a
distribution~$\mathcal{D}$ on $\{0,1\}^n$ if it produces samples from a
distribution~$\mathcal{D}'$ satisfying $c^{-1}\mathcal{D}(x) \leq
\mathcal{D}'(x) \leq c\mathcal{D}(x)$ for every~$x$.
\end{definition}

By summing the pointwise bounds, the same inequalities hold for every event.
Thus the ratio of two event probabilities is preserved up to a factor~$c^2$.
Throughout, a sampler runs in time polynomial in the circuit size, and its
tolerance $c = c(n)$ may depend on the number of qubits~$n$.

As we will see, our goal is to collapse the polynomial hierarchy to its third level under certain assumptions.  To do so, it suffices to place
$\PostBQP$ inside $\PostBPP$, and to apply three standard results.  First, Toda's
theorem places $\PH$ in $\mathsf{P}^{\PP}$~\cite{toda1991pp}, and 
secondly Aaronson's identity $\PostBQP =
\PP$~\cite{aaronson2005quantum} expresses every $\PP$ computation as a
postselected quantum circuit.
The containment of $\PostBQP$ in $\PostBPP$ therefore implies $\mathsf{P}^{\PP}$ in $\mathsf{P}^{\PostBPP}$.
Finally, $\PostBPP$ can be seen to coincide with the threshold class $\mathsf{BPP}_{\mathrm{path}}$,
which Han, Hemaspaandra, and Thierauf place in
$\mathsf{P}^{\Sigma_2[\log]}$~\cite[Theorem~3.11]{hanhema1997}, hence in
$\Delta_3$.

Throughout, postselected quantum circuits consist
of gates from $\mathscr{G} = \{H, T^\dagger, CZ\}$, a universal
set~\cite{aharonov2003simple} for which the identity $\PostBQP = \PP$ holds.
The exact gate set matters, because postselection on an event of small enough
probability arbitrarily magnifies the residual amplitude error of any
approximate compilation, such as one produced by the Solovay--Kitaev
algorithm~\cite{solovay_kitaev}.  

The following lemma derives the collapse of the polynomial hierarchy to its third level from a classical-sampling hypothesis and applies to any
circuit family that admits an exact linear-size compilation of postselected
quantum computation.  The degree-$3$ QAOA and IQP statements follow as immediate
instances.

\begin{lemma}\label{lem:template}
Let $\mathcal{C}$ be a family of postselected circuits such that every
$\PostBQP$ computation of size $T$ compiles in polynomial time into a single
circuit of $\mathcal{C}$ on $O(T)$ qubits with the same postselected output
distribution.  If every $n$-qubit circuit of $\mathcal{C}$ can be sampled
classically in polynomial time within multiplicative error $c(n) \leq 2^{n^{s}}$
for some constant $0 \leq s < 1$, then $\PH \subseteq \Delta_3$.
\end{lemma}

\begin{proof}
It suffices to prove $\PP \subseteq \PostBPP$.  To do so, let $L \in \PP$ and
let $w \in \{0,1\}^m$ be an input.  We construct a $\PostBPP$ computation
deciding whether $w \in L$.  First, the equivalence $\PostBQP=\PP$ gives a
postselected quantum computation for $L$ of size $t=\poly(m)$ on inputs of
length $m$.  Thus we may amplify this computation by running it independently
$k$ times on $w$, postselecting on all $k$ postselection events, and taking the
majority output.  The resulting postselected computation has size $O(kt)$ and,
for some constant $a>0$, independent of $m$ and $k$, conditional error at most
$2^{-ak}$~\cite[\S3]{aaronson2005quantum}.

By the hypothesis that every $\PostBQP$ computation compiles into $\mathcal{C}$,
the amplified computation becomes a circuit $Q \in \mathcal{C}$ on $n =
O(kt)$ qubits with the same postselected output distribution.  Let $p_b$ be the
probability that $Q$ returns $b$ on the output qubit and all zeros on the
postselection register.  The $\PostBPP$ machine runs the assumed classical
sampler for $Q$, postselects on the same postselection register being all zeros,
and returns the sampled output bit.  We show that this postselected classical
computation has bounded error.

Let $\tilde p_b$ be the corresponding probability for the sampler.  If either
$p_0$ or $p_1$ vanishes, the multiplicative bounds preserve that zero and the
positivity of the other probability.  The sampled postselected computation is
then correct with probability one, so we may assume $p_0p_1>0$.  The same
mutiplicative bounds give $\tilde p_0\tilde p_1>0$, so the sampler's odds are 
well defined.

Let $R=p_1/p_0$ be the true postselected odds, and let $\tilde R=\tilde
p_1/\tilde p_0$ be the sampler's postselected odds.  Our amplified computation
has conditional error at most $2^{-ak}$, so $R \geq 2^{ak-1}$ when $w \in L$ and
$R \leq 2^{-ak+1}$ when $w \notin L$, for all sufficiently large $k$.  The
multiplicative bounds distort the odds by at most a factor of $c(n)^2$, namely
\[
\frac{R}{c(n)^2} \leq \tilde R \leq c(n)^2 R.
\]

We now set $k=\lceil t^r \rceil$ for a constant $r > s/(1-s)$.  Then, $n=O(kt)$.
This gives $n^s=O(t^{s(r+1)})=o(t^r)=o(k)$, and hence $c(n)^2 \leq
2^{2n^s}=2^{o(k)}$.  The sampler's multiplicative error is therefore absorbed by
the amplified odds gap, so $\tilde R>2$ when $w \in L$ and $\tilde R<1/2$ when
$w \notin L$, for all sufficiently large input lengths.  Since odds above $2$
correspond to conditional acceptance probability above $2/3$, and odds below
$1/2$ correspond to conditional acceptance probability below $1/3$, these two
inequalities give bounded error.  These two inequalities give bounded error
when the input length is sufficiently large.  If that is not the case, then 
the inputs may be hardwired to give a valid $\PostBPP$ computation for $L$.
Thus $\PP \subseteq \PostBPP$, and therefore
\[
\PH \subseteq \mathsf{P}^{\PP}
   \subseteq \mathsf{P}^{\PostBPP}
   \subseteq \mathsf{P}^{\Delta_3}
   = \Delta_3
\]
by Toda's theorem and $\PostBPP=\mathsf{BPP}_{\mathrm{path}}\subseteq
\mathsf{P}^{\Sigma_2[\log]}\subseteq\Delta_3$.  Consequently, the polynomial
hierarchy collapses to its third level.
\end{proof}

\begin{remark}
Tracking the odds~$R = p_1/p_0$ improves on the direct argument of Bremner,
Jozsa, and Shepherd~\cite{bremner2011classical}, which tracks the conditional
acceptance probability and requires $c < \sqrt{2}$.
\end{remark}

We now fix the postselected QAOA class used in the degree-restricted statement.

\begin{definition}\label{def:postqaoa}
A language $L$ is in $\PostQAOA_p^{\deg \leq d}$ if there exist a constant
$\varepsilon$ with $0 < \varepsilon < \tfrac{1}{2}$ and a uniform family of
depth-$p$ QAOA circuits of interaction degree at most~$d$, each equipped with a
designated output qubit and a postselection register, such that for every
input~$w$:
\begin{enumerate}[label=(\roman*)]
\item the probability of the postselection register returning $\ket{0\cdots 0}$
      is nonzero;
\item if $w \in L$, the output qubit is~$1$ with probability at least $1 -
      \varepsilon$, conditioned on postselection;
\item if $w \notin L$, the output qubit is~$0$ with probability at least $1 -
      \varepsilon$, conditioned on postselection.
\end{enumerate}
We write $\PostQAOA_p = \bigcup_d \PostQAOA_p^{\deg \leq d}$ for the class with
no degree restriction.
\end{definition}

\section{Hardness at interaction degree three}\label{sec:hardness}

Two prior results establish sampling hardness for restricted circuit families
using similar postselection arguments.  Bremner, Jozsa, and
Shepherd~\cite{bremner2011classical} proved $\PostIQP = \PP$, where IQP circuits
have the form $H^{\otimes n} D\, H^{\otimes n}\ket{0^n}$ with $D$ diagonal in
the computational basis.  Farhi and Harrow~\cite{farhi2016quantum} extended this
to QAOA by proving $\PostQAOA_1 = \PostBQP$.  Both constructions reduce an
arbitrary $\PostBQP$ circuit to the restricted family by replacing each
intermediate Hadamard with a postselected diagonal coupling to a fresh auxiliary
qubit.  The two constructions differ only in the single-qubit gate $F$ that
completes each step, and hence in the diagonal coupling it requires.

We construct the diagonal coupling for general $F$ and add a preprocessing step
that reduces the interaction degree from unbounded to $3$.  That is, we prove
$\PostQAOA_1^{\deg\leq 3} = \PostBQP$ and obtain $\PostIQP^{\deg\leq 3} = \PP$
as a corollary, writing $\PostIQP^{\deg\leq d}$ for the subclass of $\PostIQP$
in which $D$ has interaction degree at most $d$.

Our general construction will transform the input circuit so that intermediate
Hadamard gates are replaced by postselected diagonal couplings, and all
remaining pre-mixer gates commute and can be collected into a single QAOA phase
layer.  That is, every gate of $\mathscr{G}$ except the Hadamard is diagonal, 
and diagonal
operators on $n$ qubits form a commutative subalgebra, so any product of such
operators remains diagonal and merges into one phase layer.  The Hadamard lies
outside this subalgebra, so each occurrence must be replaced by a diagonal
coupling to an auxiliary qubit.  To replace a Hadamard on qubit $j$, we prepare
an auxiliary $a$ in $\ket{+}$ and apply a diagonal two-qubit gate $W$ to
$(a,j)$, then apply $F$ to $j$ and postselect $j$ on $\ket{0}$.  Consequently,
qubit $a$ carries the logical wire, and we want to choose $W$ so that the map
from $j$ to $a$ acts as a Hadamard.  Figure~\ref{fig:deg}b illustrates the
coupling substitution.  For brevity, let us write $w_{ab} = \bra{a,b}W\ket{a,b}$
and $r_b = \bra{0}F\ket{b}$.  The map from $j$ to $a$ is then $V \colon
\mathbf{C}^2_j \to \mathbf{C}^2_a$ defined by
\[
V\ket{b}_j = \bra{0}_j F_j W_{a,j}
\bigl(\ket{+}_a \otimes \ket{b}_j\bigr)
\]
and has entries $V_{ab} = w_{ab} r_b/\sqrt{2}$.  The condition $V = \lambda H$
determines $W$ entrywise.

\begin{proposition}\label{prop:gadget}
Suppose $r_0 r_1 \neq 0$.  Then $V = \lambda H$ if and only if
\begin{equation}\label{eq:gadget_sol}
w_{ab} = \frac{\lambda\,(-1)^{ab}}{r_b}.
\end{equation}
A unitary diagonal $W$ with this property exists if and only if $|r_0| = |r_1|$,
and is then unique up to global phase.
\end{proposition}

\begin{proof}
If $w_{ab}$ is given by~\eqref{eq:gadget_sol}, then $V_{ab} =
\lambda(-1)^{ab}/\sqrt{2}$, so $V = \lambda H$.  Conversely, $V = \lambda H$
forces $w_{ab}\,r_b = \lambda(-1)^{ab}$, which is~\eqref{eq:gadget_sol}.  For
the unitarity claim, $|w_{ab}| = 1$ for all $a,b$ holds if and only if
$|\lambda| = |r_0| = |r_1|$.  Only the phase of $\lambda$ remains free, and it
acts as a global phase on $W$.
\end{proof}

\begin{remark}\label{rem:recover_gadgets}
With $\lambda = 1/\sqrt{2}$, the Bremner--Jozsa--Shepherd substitution $W = CZ$
corresponds to $F = H$, and the Farhi--Harrow substitution $W = \diag(1,i,1,-i)$
to $F = \tilde{H} = e^{-i\pi X/4}$.
\end{remark}

\begin{figure}
    \centering

    \begin{tikzpicture}[transform shape, scale=0.75]
        \begin{yquant}
            qubit {$\ket{\psi}_j$} q;

            hspace {2mm} q;
            box {$D_1$} q;
            hspace {2mm} q;
            h q;
            h q;
            hspace {2mm} q;
            box {$D_2$} q;
            hspace {2mm} q;
            h q;
            h q;
            hspace {2mm} q;
            box {$D_3$} q;
            hspace {2mm} q;
        \end{yquant}
    \end{tikzpicture}

    \smallskip
    {(a) Insertion of $I = H^2$ between consecutive diagonal gates.}

    \bigskip

    \begin{tikzpicture}[scale=0.8]
        \begin{yquant}[register/separation=0mm]
            qubit {$\ket{\psi}_{j-2}$} qA;
            qubit {$\ket{\psi}_{j-1}$} qB;
            qubit {$\ket{+}_a$} aux;
            qubit {$\ket{\psi}_{j}$}   qC;
            qubit {$\ket{\psi}_{j+1}$} qD;
            qubit {$\ket{\psi}_{j+2}$} qE;

            hspace {0.5cm} -;
            text {$\cdots$} qA, qB, qC, qD, qE;
            hspace {0.1cm} -;
            box {$W$} (aux, qC);
            box {$\tilde{H}$} qC;
            text {$\bra{0}$} qC;
            discard qC;
            hspace {0.1cm} -;
            text {$\cdots$} qA, qB, aux, qD, qE;
            hspace {0.5cm} -;
        \end{yquant}
    \end{tikzpicture}

    \smallskip
    {(b) The coupling, with the auxiliary drawn adjacent to~$j$ for clarity.}

    \bigskip

    \begin{tikzpicture}[transform shape, scale=0.6]
        \begin{yquant}
            qubit {$\ket{+}_{a_6}$} a6;
            qubit {$\ket{+}_{a_5}$} a5;
            qubit {$\ket{+}_{a_4}$} a4;
            qubit {$\ket{+}_{a_3}$} a3;
            qubit {$\ket{+}_{a_2}$} a2;
            qubit {$\ket{+}_{a_1}$} a1;
            qubit {$\ket{\psi}_j$} q;

            hspace {2mm} q;
            box {$D_1$} q;
            hspace {2mm} q;

            box {$W$} (a1, q);
            box {$\tilde{H}$} q;
            text {$\bra{0}$} q;
            discard q;

            hspace {2mm} a1;
            box {$W$} (a2, a1);
            box {$\tilde{H}$} a1;
            text {$\bra{0}$} a1;
            discard a1;

            hspace {2mm} a2;
            box {$D_2$} a2;
            hspace {2mm} a2;

            box {$W$} (a3, a2);
            box {$\tilde{H}$} a2;
            text {$\bra{0}$} a2;
            discard a2;

            hspace {2mm} a3;
            box {$W$} (a4, a3);
            box {$\tilde{H}$} a3;
            text {$\bra{0}$} a3;
            discard a3;

            hspace {2mm} a4;
            box {$D_3$} a4;
            hspace {2mm} a4;

            box {$W$} (a5, a4);
            box {$\tilde{H}$} a4;
            text {$\bra{0}$} a4;
            discard a4;

            hspace {2mm} a5;
            box {$e^{i\pi Z/4}$} a5;
            hspace {2mm} a5;

            box {$W$} (a6, a5);
            box {$\tilde{H}$} a5;
            text {$\bra{0}$} a5;
            discard a5;

            hspace {2mm} a6;
            box {$\tilde{H}$} a6;
        \end{yquant}
    \end{tikzpicture}

    \smallskip
    {(c) Full substitution on a wire with three diagonal operators, including
    the end-of-wire step.}

    \caption{Degree reduction for three diagonal operators on a
    single wire.}
    \label{fig:deg}
\end{figure}

\begin{theorem}\label{thm:degree3}
For every language $L \in \PostBQP$, there exists a uniform family of
postselected depth-$1$ QAOA circuits of interaction degree at most~$3$ that
decides $L$ with bounded error.  Equivalently,
\[
\PostQAOA_1^{\deg\leq 3} = \PostBQP.
\]
Moreover, the underlying compilation runs in polynomial time, maps a
postselected circuit of size $T$ to a circuit on $O(T)$ qubits, and preserves
the postselected output distribution exactly.
\end{theorem}

\begin{proof}
The inclusion $\PostQAOA_1^{\deg\leq 3} \subseteq \PostBQP$ is immediate.  For
the reverse, we will show that applying our coupling to each Hadamard converts
an arbitrary $\PostBQP$ circuit into depth-$1$ QAOA, and then bound the
interaction degree.  For $F = e^{-i\beta X}$ the unitarity condition $|r_0| =
|r_1|$ of Proposition~\ref{prop:gadget} reads $|\cos\beta| = |\sin\beta|$, and
the principal solution $\beta = \pi/4$ gives $W = \diag(1,i,1,-i)$.  One
verifies that for all states $\ket{\phi_0}$ and $\ket{\phi_1}$,
\[
\bra{0}_j \tilde{H}_j W_{a,j}
\bigl(\ket{+}_a \otimes
(\ket{0}_j\ket{\phi_0} + \ket{1}_j\ket{\phi_1})\bigr)
= \tfrac{1}{\sqrt{2}}\, H_a
\bigl(\ket{0}_a\ket{\phi_0} + \ket{1}_a\ket{\phi_1}\bigr).
\]
This identity replaces each intermediate Hadamard by a diagonal coupling to an
auxiliary, but applying it directly leaves the degree unbounded, since
consecutive diagonal gates on one wire all contribute edges to the same qubit.
We preprocess the circuit to enforce the invariant that consecutive diagonal
gates on each wire are separated by a Hadamard.

If a wire begins with a Hadamard, absorb it into the preparation $\ket{+} =
H\ket{0}$.  If it begins with a diagonal gate, prepend $H^2$ and absorb the
first $H$ into the preparation.  We separate consecutive diagonal gates by
inserting $I = H^2$, as in Figure~\ref{fig:deg}a. At the end of each wire,
insert $I = \tilde{H}\tilde{H}^\dagger$, absorb $\tilde{H}$ into the mixing
layer, and decompose $\tilde{H}^\dagger = H\, e^{i\pi Z/4}\, H$, producing two
Hadamards separated by a $1$-local diagonal.  After preprocessing, every wire
begins in $\ket{+}$ and ends with a single $\tilde{H}$.  The Hadamards partition
the wire into segments, the maximal intervals containing no Hadamard.  The first
segment begins at the state preparation, the last reaches the end of the wire,
and each segment now contains at most one diagonal gate.  After substitution the
wire occupies a distinct qubit during each segment, so the segments of a wire
correspond bijectively to the qubits of its chain in Figure~\ref{fig:deg}c.

With the invariant in place, we replace each intermediate Hadamard by the
coupling of Figure~\ref{fig:deg}b. After all substitutions, every qubit starts
in $\ket{+}$ and receives exactly one $\tilde{H}$ from the mixing layer.  All
remaining gates before the mixer are diagonal, with phases that are eighth roots
of unity.  Setting $\gamma = \pi/4$ therefore collects them into a single
integer-valued cost function $C$.  The compiled circuit is
\[
\tilde{H}^{\otimes n}\, e^{-i\tfrac{\pi}{4} C}\, \ket{+}^{\otimes n}
\]
with postselection on some qubits, which is depth-$1$ QAOA with $\beta = \gamma
= \pi/4$.

No gate acts on a qubit after its projection onto $\ket{0}$, so the
postselections commute with all subsequent operations and form a single
register, and each coupling contributes a factor of $1/\sqrt{2}$ that cancels in
the postselected distribution.  The output distribution equals that of the
original circuit.  Each intermediate Hadamard adds one auxiliary qubit and
preprocessing adds $O(1)$ gates per gate, so a size-$T$ circuit compiles into
one on $O(T)$ qubits in polynomial time.

It remains to bound the degree.  Along each wire the active qubits form a chain
of couplings, as in Figure~\ref{fig:deg}c. Each coupling moves the wire to a new
qubit, so each qubit participates in at most two coupling edges.  Preprocessing
ensured at most one diagonal gate per segment, so each qubit carries at most one
original $2$-local gate, contributing at most one additional edge.  The maximum
degree is $2 + 1 = 3$.
\end{proof}

\begin{corollary}\label{cor:degree3_iqp}
$\PostIQP^{\deg\leq 3} = \PP$.
\end{corollary}

\begin{proof}
Repeat the proof of Theorem~\ref{thm:degree3} with the IQP specialization $F =
H$, $W = CZ$.  The end-of-wire step simplifies, since inserting $I = H H$ and
absorbing the trailing Hadamard into the final Hadamard layer leaves no
$1$-local phase.  The compilation guarantees of Theorem~\ref{thm:degree3} carry
over, and combining with $\PostBQP = \PP$ gives the claim.
\end{proof}

\begin{corollary}\label{cor:collapse}
If depth-$1$ QAOA circuits of interaction degree at most~$3$ on $n$ qubits can
be sampled classically in polynomial time within multiplicative error
$2^{n^{s}}$ for some constant $0 \leq s < 1$, then $\PH \subseteq \Delta_3$.
\end{corollary}

\begin{proof}
By Theorem~\ref{thm:degree3}, every $\PostBQP$ computation of size $T$ compiles
into a single depth-$1$, degree-$3$ QAOA circuit on $O(T)$ qubits with the same
postselected output distribution, so Lemma~\ref{lem:template} applies.
\end{proof}

\begin{corollary}\label{cor:iqp_collapse}
If IQP circuits of interaction degree at most~$3$ on $n$ qubits can be sampled
classically in polynomial time within multiplicative error $2^{n^{s}}$ for some
constant $0 \leq s < 1$, then $\PH \subseteq \Delta_3$.
\end{corollary}

\begin{proof}
By Corollary~\ref{cor:degree3_iqp}, every $\PostBQP$ computation of size $T$
compiles into a single degree at most $3$ IQP circuit on $O(T)$ qubits with the
same postselected output distribution, so Lemma~\ref{lem:template} applies.
\end{proof}

\begin{remark}\label{rem:additive}
Additive error would not suffice in Corollary~\ref{cor:collapse}.  The
postselection event may have exponentially small probability, so an
additive-error guarantee on the joint distribution gives no useful bound on the
conditional distribution.  Extensions of such hardness to additive error have so far relied
on an average-case conjecture, as in the boson sampling and IQP
cases~\cite{aaronson2011computational, bremner2016average}.
\end{remark}

The hard instances of Theorem~\ref{thm:degree3} have a trivially optimizable
cost function.

\begin{proposition}\label{prop:optimization}
The cost function~$C$ constructed in Theorem~\ref{thm:degree3} can be made
monotone and maximized at $x = 1^n$ without altering the QAOA circuit.
\end{proposition}

\begin{proof}
Since $e^{-i(\pi/4) C}$ depends only on $C$ modulo $8$, individual diagonal
entries may be shifted by multiples of $8$ without altering the circuit.  As
multilinear polynomials in the bits, the $T^\dagger$ and $CZ$ terms are $x_j$
and $4 x_j x_k$, with nonnegative coefficients.  Writing $W = e^{-i(\pi/4) C_W}$
with diagonal $C_W = (0,6,0,2)$ on qubits $(a,j)$, the congruent diagonal
$(0,6,0,10)$ has multilinear form $6 x_j + 4 x_a x_j$, and replacing the
endpoint diagonal $(7,1)$ by the congruent $(7,9)$ gives $7 + 2 x_j$.  Every
term is then a multilinear polynomial with nonnegative coefficients, hence
nondecreasing in each bit, so $C$ is monotone and maximized at $x = 1^n$.
\end{proof}

\section{Tractability at interaction degree two}\label{sec:easy}

We now show that the degree-$3$ bound is tight.  A graph of maximum degree $2$
decomposes into paths, cycles, and isolated vertices, and the QAOA circuit on
each component has small cut width.

\begin{definition}\label{def:cutwidth}
A gate on qubits~$S$ crosses cut~$i$ if there exist $j,j'\in S$ such that $j \le i < j'$.  Single-qubit gates cross no cut.  The \emph{cut width} of a
circuit on qubits~$[n]$ is the maximum number of gates crossing any single cut.
\end{definition}

Markov and Shi simulate a circuit by contracting its tensor network along a
linear qubit ordering, at cost exponential in the number of gates crossing any
single cut.

\begin{proposition}[Markov--Shi {\cite[Proposition~5.1]{markov2008simulating}}]
\label{prop:markov_shi}
A quantum circuit of size~$T$ and cut width~$r$, in which each gate acts on
$O(1)$ qubits and has equal numbers of input and output qubits, can be simulated
deterministically in time $T^{O(1)}\exp[O(r)]$.
\end{proposition}

By simulation we mean computing $\Pr[Z_S = z_S]$ for a specified subset $S
\subseteq [n]$ and outcome $z_S$.  

\begin{theorem}\label{thm:degree2}
For every depth-$p$ QAOA circuit on $n$ qubits with interaction degree at
most~$2$, every $S \subseteq [n]$, and every $z_S \in \{0,1\}^{|S|}$, the
marginal probability $\Pr[Z_S = z_S]$ can be computed deterministically in time
$(np)^{O(1)}\exp(O(p))$.
\end{theorem}

\begin{proof}
Phase gates couple only qubits within the same component and mixer gates act on
single qubits, so the state factors as $\ket{\psi} = \bigotimes_c \ket{\psi_c}$.
It therefore suffices to simulate each component independently.

A graph of maximum degree~$2$ decomposes into isolated vertices, paths, and
cycles.  An isolated vertex sees only single-qubit gates, so its output
distribution is computable in $O(p)$ operations.  For a linear ordering of the
$k$ vertices of a component, let $\partial(i)$ denote the number of
interaction-graph edges crossing cut~$i$.  Within each phase layer the terms on
a pair are commuting diagonal gates and merge into a single two-qubit gate, so
each edge contributes one gate per layer and the circuit cut width is $p \cdot
\max_i \partial(i)$.  For a path in the natural ordering, each edge $\{i,i{+}1\}$
crosses only cut~$i$, so $\max_i \partial(i) = 1$ and the cut width is~$p$.  For a
cycle, the wrap-around edge $\{1,k\}$ crosses every cut, so $\max_i \partial(i) =
2$ and the cut width is~$2p$.

In every case the cut width is $O(p)$ and the circuit has size $O(kp)$.
Proposition~\ref{prop:markov_shi} gives runtime $(kp)^{O(1)}\exp(O(p))$ per
component, and summing over components of total size~$n$ gives
$(np)^{O(1)}\exp(O(p))$.
\end{proof}

\begin{corollary}\label{cor:sampling}
The output distribution of a depth-$p$ QAOA circuit on $n$ qubits with
interaction degree at most~$2$ can be sampled exactly by a classical randomized
algorithm in time $(np)^{O(1)}\exp(O(p))$.  When $p = O(\log n)$, this is
$n^{O(1)}$.
\end{corollary}

\begin{proof}
Sample $z_1, \ldots, z_n$ sequentially by the chain rule.  At step $t$, two
calls to Theorem~\ref{thm:degree2} give $\Pr[Z_1 = z_1, \ldots, Z_{t-1} =
z_{t-1},\, Z_t = b]$ for $b \in \{0,1\}$, and we draw $z_t$ from the conditional
distribution.  After $n$ steps, $z = z_1 \cdots z_n$ is drawn from the exact
output distribution.  The $2n$ calls cost $(np)^{O(1)}\exp(O(p))$ total, which
is $n^{O(1)}$ when $p = O(\log n)$.
\end{proof}

\section{Discussion}\label{sec:discussion}

We see three natural directions for further work.  First, a quantitative
question remains on each side of the threshold.  The degree-$2$ simulation runs
in time exponential in~$p$, and whether this dependence can be removed is open.
The collapse argument, in turn, tolerates multiplicative error $2^{n^{s}}$ for
every fixed $s < 1$ but not at the supremum, and whether hardness persists at
error $2^{cn}$ for some $c > 0$ is open.

Second, our hardness results are proved for multiplicative error, and the
additive-error analogue remains open.  Existing arguments for additive error, as
in the boson sampling and IQP cases~\cite{aaronson2011computational,
bremner2016average}, pair anticoncentration of a random instance family with a
conjecture that approximating its output probabilities is \#P-hard on average.
For depth-$1$ QAOA with $2$-local cost functions, both would concern random cost
functions on degree-$3$ graphs.  The worst-case form of this hardness does
follow from our construction, since a multiplicative approximation of the
compiled output probabilities recovers the postselected odds and hence decides
$\PP$.  The closest IQP results~\cite{bremner2016average} do not transfer.
Their anticoncentration proof needs a random weight on every pair of qubits and
even their sparsest average-case families~\cite{bremner2017achieving} have
degree $\Theta(\log n)$.

Finally, the instances behind our hardness results have trivially optimizable
cost functions, so their sampling advantage carries no optimization advantage.
We would therefore like a characterization of the instance families in which the
output distribution is hard to sample and the cost function is nontrivial to
optimize.  Any optimization advantage of QAOA must occur in such a family, since
a classical sampler for the output distribution replicates the optimization
performance.

\subsection*{Acknowledgment}
This research was supported by QuantERA project HQCC (Hybrid Quantum Classical
Computation).

\bibliographystyle{amsalpha}
\bibliography{references}

\end{document}